\documentclass[runningheads]{llncs}
 \usepackage{fullpage}
\usepackage{macros}
\usepackage[english]{babel}
\usepackage[utf8]{inputenc} 

\usepackage[T1]{fontenc}    
\renewcommand{\vec}[1]{\boldsymbol{#1}}
\title{Hardness Amplification for (Sparse) LPN}
\author{Divesh Aggarwal \and Rishav Gupta \and Li Zeyong}
\institute{National University of Singapore, \texttt{divesh@comp.nus.edu.sg}
		\and 
		National University of Singapore, \texttt{rishavg@u.nus.edu}
        \and
        National University of Singapore, \texttt{li.zeyong@u.nus.edu}}

\date{}

\begin{document}

\maketitle
\begin{abstract}

  We prove new hardness amplification results for Learning Parity with Noise (\textsf{LPN}) and its sparse variants. In $\mathsf{LPN}_{\eta,n,m}$, the goal is to recover a secret $\vec s\in\mathbb{F}_2^n$ from $m$ noisy linear samples $(\vec a,b)$, where $\vec a\leftarrow \mathbb{F}_2^n$ is uniform and $b=\langle \vec a,\vec s\rangle + e$ with $e\leftarrow \mathrm{Ber}(\eta)$.

Building on the direct-product framework introduced by Hirahara and Shimizu~\cite{HS23}, we show an \emph{instance-fraction amplification} theorem: for any $\varepsilon,\delta>0$, any algorithm that solves $\mathsf{LPN}_{\eta,n,m}$ with success probability $\varepsilon$ can be transformed into an algorithm that succeeds with probability $1-\delta$ on a related \textsf{LPN} distribution with scaled parameters $\mathsf{LPN}_{\eta/k,\;n/k,\;m}$, where
$
k=\Theta\!\left(\frac{1}{\delta}\log\frac{1}{\varepsilon}\right).
$
Equivalently, an algorithm that solves \textsf{LPN} on a {\em small fraction of instances} can be converted into an algorithm that solves \textsf{LPN} on {\em almost all instances}, yielding a self-amplification for a wide range of parameters.

We extend the same amplification approach to \textsf{LPN} over $\mathbb{F}_q$ and to Sparse-$\mathsf{LPN}$, where each query vector $\vec a$ has exactly $\sigma$ nonzero entries. Together, these results establish hardness self-amplification for a broad family of \textsf{LPN}-type problems, strengthening the foundations for assuming the average-case hardness of \textsf{LPN} and its sparse variants.
  
\end{abstract}
\section{Introduction}
\paragraph{Learning Parity with Noise $(\lpn)$.}
Learning Parity with Noise $(\lpn)$ is a noisy linear problem over $\mathbb{F}_2$ and can be viewed as a randomized form of syndrome decoding. In the standard \emph{search} version, one is given a matrix $\vec{A}\in\mathbb{F}_2^{m\times n}$ and a vector
$
\vec{b}=\vec{A}\vec{s}+\vec{e}\in\mathbb{F}_2^m,
$
where the secret $\vec{s}\leftarrow \mathbb{F}_2^n$ is uniform and the noise vector $\vec{e}\leftarrow \operatorname{Ber}(\eta)^m$ has independent Bernoulli coordinates. The goal is to recover $\vec{s}$ from $(\vec{A},\vec{b})$. Since the early 1990s, $\lpn$ and its variants have been used as average-case hardness assumptions in cryptography and learning theory~\cite{BFKL93,BKW03}.

Many constructions rely on $\lpn$ in different parameter regimes. In the constant-noise regime (typically $m=\mathrm{poly}(n)$ and constant $\eta$), it gives light weight shared secret authentication protocols such as~\cite{HB01,JW05}. In lower-noise regimes (often $m=O(n)$ and $\eta=O(n^{-1/2})$), it implies public-key encryption, including dual-Regev-style schemes~\cite{Ale03} and CCA-secure variants~\cite{DMN12,KMP14}. $\lpn$ has also been used for pseudorandom generators and one-way functions~\cite{Pie12}, symmetric-key encryption with adaptive IND-CPA security~\cite{GRS08}, and commitments and zero-knowledge proofs~\cite{JKP12} (see also~\cite{YZW19} for hashing in lower-noise regimes). A practical feature of $\lpn$ is efficiency: operations over $\mathbb{F}_2$ correspond to XOR/AND computations, and Ring-$\lpn$ variants have been studied for efficient authentication and encryption~\cite{HKL12,YGZ18}.

\paragraph{Sparse Learning Parity with Noise.}
Sparse-$\lpn$ is a variant of $\lpn$ where the challenge matrix $\mathbf{A}\in\mathbb{F}_2^{m\times n}$ is \emph{row-sparse}: each row has Hamming weight at most (often exactly) $\sigma=O(1)$. The samples have the form $u=\mathbf{A}\mathbf{s}+\mathbf{e}$ for a uniform secret $\mathbf{s}\in\mathbb{F}_2^n$ and noise $\mathbf{e}\leftarrow \operatorname{Ber}(\eta)^m$ (typically with inverse-polynomial $\eta$). Sparse-$\lpn$ has been used for linear-stretch pseudorandom generators in $\mathrm{NC}^0$~\cite{AIK08}, public-key encryption from related sparse-XOR regimes~\cite{ABW10}, and a number of cryptographic primitives based on noisy local linear constraints, \ ~\cite{IKOS08,BCGI18,BCG+20,BCG+22,CRR21,RRT23,DIJL23,BCCD23}. Sparse $\lpn$ (and the more general Local Functions with Noise model) assumption plays a key role in recent work on indistinguishability obfuscation~\cite{JLS21,JLS22}.

\paragraph{Worst-case to average-case reductions.}
For the closely related Learning with Errors (LWE) problem, worst-case-to-average-case reductions from lattice problems are known and form a main part of its foundations~\cite{Reg05,Pei09,BLPRS13}. For $\lpn$, the known worst-case-to-average-case results are weaker.

The first such reductions for variants of $\lpn$ use \emph{code smoothing}. Brakerski, Lyubashevsky, Vaikuntanathan and Wichs~\cite{BLVW19} gave Fourier-analytic smoothing lemmas for binary linear codes, reducing worst-case instances of a balanced nearest codeword problem at very low error rate (about $\tfrac{\log n}{n}$) to average-case $\lpn$ at very high noise (about $\tfrac{1}{2}-\tfrac{1}{n^c}$). Yu and Zhang~\cite{YZ21} extended this approach and obtained implications for constant-noise $\lpn$ under stronger nearest-codeword assumptions. Debris-Alazard and Resch~\cite{DR22} developed a more general Fourier-analytic framework and clarified limitations of smoothing bounds.

A common feature of these results is that they start from worst-case regimes with extremely low relative error, for example corruption rate $t/m = O(\log n/n)$. In this regime, one can solve the worst-case instance in quasipolynomial time by sampling $r=\Theta(n)$ equations, enumerating the $O(\log n)$ corrupted ones (which holds with high probability), and then applying Gaussian elimination. Thus, these reductions start from worst-case sources that are far from the parameters used in cryptography, and they do not rule out quasipolynomial-time algorithms for average-case $\lpn$ in cryptographic regimes (e.g., constant noise or $\eta=1/n^c$).

These works leave open the following natural question.
\begin{question}[Next best to worst-case-to-average-case for $\lpn$]\em
Can one prove an ``almost'' worst-case-to-average-case statement for \emph{Search-LPN} in cryptographic regimes: namely, that hardness on a small $\delta$ fraction of random instances can be amplified to hardness on almost all random instances (and hence average-case hardness), under a standard parameter transformation?
\end{question}

\paragraph{Sparse-$\lpn$ and worst-case to average-case.}
For Sparse-$\lpn$, the situation is even less developed. Unlike LWE, and unlike the partial progress for $\lpn$ via code smoothing, there are currently no known worst-case-to-average-case reductions for Sparse-$\lpn$ in cryptographic regimes. Moreover, existing approaches based on smoothing and extractor-style ``random hashing'' do not seem to apply in any useful way: smoothing arguments rely on properties of dense random codes, while sparse instances correspond to highly structured, low-density constraints; and applying sparse hashing tends to destroy sparsity or forces parameters into regimes that are algorithmically easy.

This leaves open the following analogue of the question above.
\begin{question}[Next best to worst-case-to-average-case for Sparse-$\lpn$]\em
Can one prove an ``almost'' worst-case-to-average-case statement for \emph{Search Sparse-$\lpn$} in cryptographic regimes: namely, that hardness on a small $\delta$ fraction of random Sparse-$\lpn$ instances can be amplified to hardness on almost all instances (and hence average-case hardness), under a standard parameter transformation that preserves sparsity?
\end{question}

\paragraph{Our perspective: self-amplification.}
We take a different route and work directly with the standard average-case $\lpn$ distribution. Our results are a form of \emph{instance-fraction amplification} for \emph{Search-LPN}: assuming that $\lpn$ is hard on even a small fraction of random instances implies hardness on almost all instances (and hence average-case hardness) at related parameters. Roughly, given an algorithm that recovers the secret on random instances with success probability $\varepsilon$, we build an algorithm that recovers the secret with success probability $1-\delta$ on a $1-\delta$ fraction of instances, for small $\delta$ (with an explicit parameter change). A full worst-case-to-average-case reduction would correspond to the special case $\delta=0$.

Our proof uses the self-amplification framework introduced by Hirahara and Shimizu~\cite{HS23}. We instantiate their framework for $\lpn$ by showing how to combine $k$ independent $\lpn$ instances into a single larger instance with an explicitly transformed noise rate, and then applying the HS23 analysis together with a verification test that succeeds on all but a negligible fraction of instances.

\paragraph{Hirahara and Shimizu's approach.}
Hirahara and Shimizu~\cite{HS23} consider a distributional problem $(f,\mu)$ and its $k$-fold product $(f^{\times k},\mu^{\times k})$. The standard direct-product reduction places an input $x\leftarrow \mu$ into a random coordinate of a $k$-tuple and fills the other coordinates with fresh samples from $\mu$, then queries a solver for the product problem. The question is whether this procedure samples from the solver's easy inputs often enough to give high success on most base inputs. HS23 show that for suitable $k$ (typically $k=\Theta({1}/{\delta} \cdot \log(1/\varepsilon))$), this amplifies success on the product problem into much higher success on the base problem for a $1-\delta$ fraction of inputs.

\subsection{Our contributions and techniques}

\noindent\textbf{From HS23 to $\lpn$.}
We use the hardness self-amplification framework of Hirahara and Shimizu~\cite{HS23}. In our setting, the underlying search problem is Search-$\lpn$: the function $f$ maps an $\lpn$ instance $(A,b)$ to its secret $s$ (and is undefined on the negligible set of non-decodable instances; we handle this using a verification/decodability test).

To apply~\cite{HS23}, we need to relate the $k$-fold product distribution $(f^{\times k},\mu^{\times k})$ to a \emph{single} $\lpn$ distribution. We do this by an explicit product-to-single embedding. Given $k$ independent $\lpn$ instances
$$
(A_i,\; b_i = A_i s_i + e_i)\qquad\text{for }i\in[k]\;,
$$
with $A_i\in \mathbb{F}_2^{m\times (n/k)}$, we form
$$
A \;:=\; [A_1\mid A_2\mid \cdots \mid A_k]\in \mathbb{F}_2^{m\times n},
\qquad
b \;:=\; \bigoplus_{i=1}^k b_i \;=\; A s \oplus e\;,
$$
where $s:=(s_1,\ldots,s_k)\in \mathbb{F}_2^n$ and $e:=\bigoplus_{i=1}^k e_i$. The matrix $A$ is uniform over $\mathbb{F}_2^{m\times n}$, and the noise distribution transforms exactly as
$$
e \sim \operatorname{Ber}(\eta')^m
\qquad\text{with}\qquad
\eta' \;=\; \frac{1-(1-2\eta)^k}{2}\;.
$$
Thus, any solver for a single $\lpn$ instance at parameters $(n,\eta')$ yields a solver for $k$ independent $\lpn$ instances at parameters $(n/k,\eta)$ with the same success probability. Combining this embedding with the sampler analysis of~\cite{HS23} gives our amplification theorems for $\lpn$ (and similarly over $\mathbb{F}_q$).

\vspace{0.5em}
\noindent\textbf{Sparse-$\lpn$.}
For Sparse-$\lpn$, the same embedding suggests that sparsity should scale by a factor of $k$ when we concatenate blocks. The difficulty is that the usual Sparse-$\lpn$ distribution fixes the row weight to be exactly $\sigma$, and this distribution is not closed under concatenation into blocks.

We handle this by introducing an \emph{approximate-sparse} variant in which each entry of $A$ is independently $1$ with probability $\sigma/n$ (so each row has expected weight $\sigma$). This distribution \emph{is} closed under block concatenation: concatenating $k$ independent matrices of dimension $n/k$ with per-entry probability $(\sigma/k)/(n/k)=\sigma/n$ again gives an $n$-column matrix with per-entry probability $\sigma/n$. We prove self-amplification for this approximate-sparse variant using the same product-to-single embedding and~\cite{HS23}.

To return to the standard Sparse-$\lpn$ model (exact row weight $\sigma$), we use a simple filtering reduction: from an approximate-sparse sample, keep only those rows of weight exactly $\sigma$. Conditioned on a row having weight $\sigma$, its distribution is uniform over weight-$\sigma$ vectors, so the retained rows form a valid Sparse-$\lpn$ instance. This incurs only a controlled blowup in the number of samples.

\vspace{0.5em}
\noindent\textbf{$q$-ary extension.}
All steps above extend to $\mathbb{F}_q$ for prime $q$. The embedding is identical, except that the noise sum is over $\mathbb{F}_q$ and the noise parameter becomes
$
\eta' \;=\; 1-(1-\eta)^k.
$

\vspace{0.5em}
\noindent\textbf{Informal summary of results.}
Below, arrows should be read as: “a solver for the bottom problem implies a solver for the top problem,” with an explicit polynomial-time reduction.

\begin{itemize}
\item \textbf{Self-amplification for Search-$\lpn$.}
For suitable $k$ and all relevant parameters, a solver for
$
\lpn_{\eta',\,n,\,m}
$ with success $\varepsilon
$
implies a solver for
$
(\lpn_{\eta,\,n/k,\,m})^{\times k}
\quad\text{with success } \varepsilon,
$
where $\eta'=\frac{1-(1-2\eta)^k}{2}$ (binary) or $\eta'=1-(1-\eta)^k$ ($q$-ary). Using~\cite{HS23} plus verification, this yields amplification:
from success $\varepsilon$ on the product distribution we obtain success at least $1-\delta$ on a $1-\delta$ fraction of base instances (for appropriate $k$).

\item \textbf{Sparse case via approximate-sparse.}
We prove the same amplification statement for the approximate-sparse distribution (per-entry probability $\sigma/n$), and then reduce approximate-sparse to exact-sparse by filtering rows of Hamming weight exactly $\sigma$, with a mild increase in the number of samples. The sparsity rate is preserved (dimension and sparsity scale together).
\end{itemize}

\noindent We expect the same approach to apply to other noisy linear problems whose instance distributions are compatible with direct-product embeddings and admit efficient verification on almost all inputs.

\subsection{Discussion and Comparison to Prior Work}

\paragraph{Why leftover-hash style ``smoothing'' does not help for $\lpn$.}
A standard way to obtain worst-case-to-average-case implications for lattice problems is to apply an extractor-style transformation (see, e.g.,~\cite{GMPW20} in the LWE setting). The basic idea is to left-multiply by a random hash matrix $M\in\F_2^{m\times m}$ whose rows are $\sigma$-sparse, with the goal that for most $A$ the pair $(MA,A)$ is statistically close to $(U,A)$, where $U$ is uniform over $\F_2^{m\times n}$. For such a leftover-hash argument to give small statistical distance, the hash family must have sufficiently large entropy; for $\sigma$-sparse rows this requires roughly
$
\sigma \log\!\Big(\frac{n}{\sigma}\Big)\ \gtrsim\ n
\qquad\Rightarrow\qquad
\sigma\ \gtrsim\ \frac{n}{\log n}.
$
Under this transformation,
$
M(A s + e) \;=\; (MA)s \;+\; Me,
$
so the noise rate changes to
$
\eta' \;=\; \frac{1-(1-2\eta)^{\sigma}}{2}\ \approx\ \sigma\eta
\qquad\text{(for small $\eta$)}.
$
Therefore, if $\sigma \gtrsim n/\log n$ and we want $\eta'$ bounded away from $1/2$, we must have $\eta \lesssim (\log n)/n$. This is a regime where LPN is algorithmically easy when $m=\poly(n)$, since the number of corrupted equations is at most $O(\log n)$ and one can enumerate the error locations and solve by Gaussian elimination. In short, leftover-hash style smoothing pushes the parameters into a regime that is not relevant for cryptography.

For sparse-$\lpn$, the situation is worse: left-multiplication by a dense hash destroys sparsity, while restricting to sparse hashes runs into the same entropy/noise barrier above. We are not aware of a way to obtain useful cryptographic-parameter implications for sparse-$\lpn$ from such smoothing arguments.

\paragraph{Our approach and why it avoids these bottlenecks.}
Our results follow the hardness self-amplification framework of Hirahara and Shimizu~\cite{HS23}. The starting point is that solving $k$ independent instances with probability $\varepsilon$ can, under suitable conditions, be amplified to success close to $1$ on almost all instances. In the HS23 framework, this amplification is obtained by analyzing the standard direct-product reduction as a sampler on the input distribution.

We instantiate this framework for $\lpn$ using a product-to-single embedding: we combine $k$ independent $\lpn$ instances in dimension $n/k$ into one $\lpn$ instance in dimension $n$ by concatenating the challenge matrices and summing the responses. Under this embedding, the noise parameter transforms exactly as
$
\eta'_{\mathrm{bin}} \;=\; \frac{1-(1-2\eta)^k}{2}
\qquad\text{and}\qquad
\eta'_{q} \;=\; 1-(1-\eta)^k,
$
while the dimension scales by $n' = n/k$. Together with the sampler analysis of~\cite{HS23}, this gives quantitative instance-fraction amplification: for appropriate $k$ (e.g., $k=\poly(n)$), a solver with success $\varepsilon$ on the product distribution can be converted into a solver that succeeds with probability at least $1-\delta$ on a $1-\delta$ fraction of instances, with $\delta$ controlled by $k$ (see Section~\ref{sec:hs23} for the precise statement and parameters). Importantly, this amplification holds in the standard cryptographic noise regimes for $\lpn$ (constant noise down to $\eta=1/n^{\alpha}$), without forcing $\eta$ into the easy $(\log n)/n$ range.

For sparse-$\lpn$, we avoid the sparsity blow-up that arises in hashing-based approaches by going through an \emph{approximate-sparse} variant, where each entry of $A$ is $1$ independently with probability $\sigma/n$. This distribution is compatible with the product-to-single embedding. We prove self-amplification for the approximate-sparse variant, and then reduce back to the standard exact-sparse model by filtering rows of weight exactly $\sigma$. This step increases the number of samples by only a controlled factor and loses only negligible success probability, while keeping sparsity at essentially $\sigma$.

\paragraph{Relation to~\cite{HS23}.}
Hirahara and Shimizu~\cite{HS23} illustrate their framework on linear-algebraic tasks such as matrix multiplication, where amplification can be combined with a cancellation identity (e.g., expanding $(A+R)(B+S)$). For $\lpn$, such cancellation identities are not available because the unknown secret $s$ prevents us from generating the analogue of ``random masks'' of the form $R s$ with the correct noise. Our contribution is to show how to apply the HS23 framework to $\lpn$ anyway, by giving a direct-product embedding that stays within the $\lpn$ distribution (including sparse variants) and by using a verification/decodability procedure to support instance-fraction amplification in cryptographic parameter regimes. We leave open the question of obtaining a true worst-case-to-average-case reduction for $\lpn$ with parameters comparable to those achieved in our work.


\subsection{Organization of the Paper} In \cref{sec:prelim}, we recollect some notations and tools useful for our presentation. We also show that for each of the LPN variants considered, a random instance is (essentially) uniquely decodable with all but negligible probability.

In \cref{sec:hs23}, we present a detailed exposition of the self hardness amplification framework introduced in \cite{HS23}.

In \cref{sec:lpn_plain}, we prove self-hardness amplification for the standard 
$\lpn_{q,\eta,n,m}$ problem and its binary analogue $\lpn_{\eta,n,m}$, 
leveraging the direct-product framework of~\cite{HS23}. To extend these 
results to sparse-$\lpn$, we introduce a new variant $\lpn_{\eta,n,m}^{\approx \sigma}$. 
Recall that in the sparse version $\lpn_{\eta,n,m}^{= \sigma}$, each query vector 
$a \in \F_2^n$ has Hamming weight exactly $\sigma$, whereas in the 
approximate-sparse variant $\lpn_{\eta,n,m}^{\approx \sigma}$, each coordinate of $\vec a$ 
is sampled so that the expected Hamming 
weight is $\sigma$. 

In \cref{sec:sparse_section}, we establish hardness self-amplification for this approximate-sparse variant, and then give a reduction 
from $\lpn_{\eta,n,M}^{\approx \sigma}$ to the exact-sparse version 
$\lpn_{\eta,n,m}^{= \sigma}$ by discarding non-$\sigma$-sparse samples. While this 
incurs a mild blowup in sample complexity, the loss in success probability is 
only negligible. Combining these steps, we obtain a reduction from 
$\lpn_{\eta,n,M}^{\approx \sigma}$ with success $1-\delta$ to 
$\lpn_{k\eta,kn,m}^{= k\sigma}$ with success $\varepsilon$, for both the binary 
and $q$-ary variants, thereby strengthening the hardness foundations of 
$\lpn$ as well as Sparse-$\lpn$.


\section{Preliminaries}\label{sec:prelim}

We use the following notations. For a set $S$, we write $\vec{x} \sim S$ to mean that $\vec{x}$ is chosen uniformly at random from $S$. For a distribution $\mathcal{D}$,we  write $\vec{x} \sim \mathcal{D}$ to mean that $\vec{x}$ is drawn according to $\mathcal{D}$. We also use the notation $\binom{[n]}{s}$ to denote the set of all binary vectors in $\F_2^n$ of Hamming weight exactly $s$. We use the notation $\norm{\vec x}_0$ in order to denote the hamming weight of $\vec x$.

\subsection{Information Theory and Probability}
We begin by defining the $q$-ary entropy function, $H_q$.
\begin{definition}[$q$-ary entropy function]
For $\rho\in[0,1]$, the $q$-ary entropy function is defined as follows,
\[
H_q(\rho)\;:=\;\rho\log_q(q-1)\;-\;\rho\log_q\rho\;-\;(1-\rho)\log_q(1-\rho).
\]
\end{definition}
We can use the $q$-ary entropy function in order to estimate the volume of a Hamming ball.
\begin{lemma}[Entropy Bound \cite{guruswami2012essential}]
For integers $n\ge 1$ and $0\le k\le n$, let $B_q(k,n)$ denote the number of vectors in $\F_q^n$ of Hamming weight at most $k$ then,
\[
q^{\,nH_q(\rho) -o(n)}
\;\le\;
B_q( \rho n,n)
\;\le\;
q^{\,nH_q(\rho)},
\qquad\text{for }\rho\in[0,1-1/q].
\]

\end{lemma}
The following lemma is the standard Chernoff's Bound.
\begin{lemma}[Chernoff's Bounds]
Let $X = \sum_{i=1}^m X_i$, where $X_i$ are independent Bernoulli random variables, and let $\mu = \mathbb{E}[X]$. Then for any $0 < \delta < 1$,
\[
\Pr\!\left[X \leq (1-\delta)\mu \right] \;\leq\; \exp\!\left(-\frac{\delta^2}{2}\mu\right),
\]
and for any $\delta > 0$,
\[
\Pr\!\left[X \geq (1+\delta)\mu \right] \;\leq\; \exp\!\left(-\frac{\delta^2}{2+\delta}\mu\right).
\]
\end{lemma}

We  now define the Bernoulli distribution.
\begin{definition}[Bernoulli Distribution]
A boolean random variable $X$ is said to follow a \emph{Bernoulli distribution} with parameter $\eta \in [0,1]$,  
denoted $X \sim \Ber(\eta)$, if
\[
\Pr[X = 1] = \eta
\quad \text{and} \quad
\Pr[X = 0] = 1 - \eta.
\]
\end{definition}
The following lemma describes the distribution generated by sum of $k$ identical and independent Bernoulli random variables over $\F_2$.
\begin{lemma}

\label{ber_2}
Let $X_1,\dots,X_k$ be independent random variables where $X_i \sim \Ber(\eta)$,
then  the following random variable,
\[
S \;=\; X_1 \oplus X_2 \oplus \cdots \oplus X_k,
\]
follows the distribution $\Ber(\theta_k(\eta))$, where 
$\theta_k(\eta)= \frac{1-(1-2\eta)^k}{2}.
$
\end{lemma}
\begin{proof}
Let $Y_i := (-1)^{X_i}\in\{\pm1\}$. Then $\E[Y_i] = (1-\eta)\cdot 1 + \eta\cdot(-1) = 1-2\eta$. Notice that 
\[
(-1)^{S} := Y_1 \cdot Y_2 \cdots Y_k \;.
\]
Since $Y_1, \ldots, Y_k$ are independently distributed, we have that 
\[
\E\big[(-1)^S\big] \;=\; \E\Big[\prod_{i=1}^k (-1)^{X_i}\Big]
\;=\; \prod_{i=1}^k \E[(-1)^{X_i}]
\;=\; (1-2\eta)^k.
\]
But $\E[(-1)^S] = \Pr[S=0]-\Pr[S=1] = 1 - 2\,\Pr[S=1]$.  
Hence 
\[
\Pr[S=1] \;=\; \frac{1-(1-2\eta)^k}{2} \;=:\; \theta_k(\eta)\;.
\]
Therefore $S \sim \Ber(\theta_k(\eta))$.
\qed\end{proof}
We now extend the notion of the Bernoulli distribution to finite fields $\F_q$.
\begin{definition}[General Bernoulli Distribution]
    A random variable $X \in \F_q$ is said to follow the distribution $\Ber_q(\eta)$, if it samples a uniform element from $\F_q$ with probability $\eta$, otherwise returns $0$. More formally, 
\[
X \sim {\Ber}_q(\eta) \quad \text{if} \quad
\Pr[X = x] =
\begin{cases}
(1-\eta)+ \dfrac{\eta}{q} & \text{if } x = 0, \\[8pt]
\dfrac{\eta}{q} & \text{if } x \in \F_q \setminus \{0\}.
\end{cases}
\]
\end{definition}
Observe that $\Ber(\eta)$ distribution is same as $\Ber_2(2 \eta)$ distribution. The only reason we don't unify both definitions is to keep both definitions consistent with those in the literature. 

In
the following lemma we describe the distribution generated by sum of $k$ identical and independent general Bernoulli random variables.
\begin{lemma}
\label{ber_q}
Let $X_1,\dots,X_k$ be independent random variables where,

$X_i \sim \Ber_q(\eta)$ 
then the following random variable,
\[
S \;=\; X_1 + X_2 + \cdots + X_k \mod q,
\]
over $\F_q$ is distributed as $\Ber_q(\phi_k(\eta))$, where 
$\phi_k(\eta)= {1-(1-\eta)^k}$.
\end{lemma}
\begin{proof}
Use the equivalent sampling rule for $\Ber_q(\eta)$: with probability $1-\eta$ output $0$, and with probability $\eta$ output a uniform element of $\F_q$.  
Write $X_i = B_i U_i$, where $B_i \sim \Ber(\eta)$ and $U_i \sim \mathrm{Unif}(\F_q)$ are independent (and all pairs $(B_i,U_i)$ are independent across $i$). Then
\[
S \;=\; \sum_{i=1}^k B_i U_i.
\]
If all $B_i=0$ (which happens with probability $(1-\eta)^k$), then $S=0$. Otherwise, conditioned on the event that, there exists at least one $i$ such that $B_i = 1$, the random variable $S$ is $U_i + Y$, where $Y = \sum_{j \neq i} B_j U_j$ is independent of $U_i$, and hence $S$ is uniform in $\F_p$. This implies that
\[
\Pr[S=0] \;=\; (1-\eta)^k \cdot 1 \;+\; \big(1-(1-\eta)^k\big)\cdot \frac{1}{q}
\;=\; \big(1-\phi_k(\eta)\big) + \frac{\phi_k(\eta)}{q},
\]
and for each $a\neq 0$,
\[
\Pr[S=a] \;=\; \big(1-(1-\eta)^k\big)\cdot \frac{1}{q}
\;=\; \frac{\phi_k(\eta)}{q}.
\]
This shows that $S$ id distributed as $\Ber_q\!\big(\phi_k(\eta)\big)$.
\qed\end{proof}

Note, both $\theta_k(\eta)$ and $\phi_k(\eta)$ are smaller than $k{\eta}{}$. The following is the standard notion of statistical distance.


\begin{definition}[Statistical Distance]
Let $\mathcal{D}_1$ and $\mathcal{D}_2$ be two probability distributions over the same finite domain $\mathcal X$. Let $\Pr_{\mathcal D}[x]$ denote the probability that the random variable sampled from the distribution $\mathcal D$ takes the value $x$.
 
The \emph{statistical distance} (or total variation distance) between them is defined as
\[
\Delta(\mathcal{D}_1, \mathcal{D}_2) \;=\; 
\frac{1}{2}\sum_{x \in \mathcal X} \left| \Pr_{\mathcal{D}_1}[x] - \Pr_{\mathcal{D}_2}[x] \right|.
\]
\end{definition}

The following is the standard data processing inequality.
\begin{lemma}[Data Processing Inequality]\label{lem:data_processing}
Let $P,Q$ be discrete random variables on a finite set. For any (possibly randomized) function $f$,
\[
\Delta\bigl(f(P),\,f(Q)\bigr)\;\le\;\Delta(P,Q).
\]
\end{lemma}
\subsection{Stirling's Approximation}
\begin{lemma}[Stirling-Robbins Bound \cite{Robbins1955}]
\label{Rob}
For every integer $n \ge 1$,
\[
\sqrt{2\pi n}\,\Bigl(\tfrac{n}{e}\Bigr)^{\!n}\, e^{\frac{1}{12n+1}}
\;<\;
n!
\;<\;
\sqrt{2\pi n}\,\Bigl(\tfrac{n}{e}\Bigr)^{\!n}\, e^{\frac{1}{12n}}\, .
\]
\end{lemma}

\begin{corollary}[Stirling approximation for binomial coefficients] \label{lem:aprox_bion}
For integers $n \geq s \geq 1$, the binomial coefficient satisfies,
\[
\frac{1}{\sqrt{8s(1-\tfrac{s}{n})}}
  \cdot 
  \frac{n^n}{s^s (n-s)^{\,n-s}}
  \;\;\leq\;\;
  {\binom{n}{s}}
  \;\;\leq\;\;
\frac{1}{\sqrt{2\pi s(1-\tfrac{s}{n})}}
  \cdot 
  \frac{n^n}{s^s (n-s)^{\,n-s}}.
\]
\end{corollary}

We include a standard proof in \cref{appendix} for completeness.

\subsection{Variants of \texorpdfstring{$\lpn$}{LPN}}

$\lpn$ is a distributional problem, equipped with distributions $\mathcal S, \mathcal D_a, \mathcal D_e$. For a secret vector $\vec s$ sampled from distribution $\mathcal S$, we are given oracle access to samples of $(\vec a,b)$, where $\vec a \sim \mathcal D_a$ and 
$b = \langle \vec a, \vec s \rangle + e$ with $e \sim \mathcal D_e$. We can collectively view these samples as $\paren{\vec A, \vec b=\vec A\cdot\vec s + \vec e}$. The goal of the Search-$\lpn$ problem is to recover the secret $\vec s$ given $(\vec A, \vec b)$. Now by changing the distributions associated with the $\lpn$ problem we define the following variants,
\begin{enumerate}

\item This is the classical variant defined over the binary field, \begin{definition}[$\lpn_{\eta, n,m}$]
    For a uniformly random secret $\vec s \sim \F_2^n$, we are given oracle access to $m$ samples of the form $(\vec a,b)$, where 
$\vec a$ is uniform in $\F_2^n$ and 
$b = \langle \vec a, \vec s \rangle + e$ with 
$e \sim \text{Ber}(\eta)$. The goal of the problem is to recover the secret $\vec s$ using these $m$ samples.
\end{definition}

\item The following is the large field variant which generalizes this problem to $\F_q$, \begin{definition}[$\lpn_{q,\eta,n,m}$]\label{large_lpn}

For a prime $q$, and a secret $\vec s \sim \F_{q}^{n}$, we are given access to $m$ samples $(\vec a,b)$, where 
$\vec a \sim \paren{\F_q}^{n}$ uniformly and 
$b = \langle \vec a, \vec s \rangle + e$ with 
$e \sim {\Ber}_q(\eta)$. The goal of the problem is to recover the secret $\vec s$ using these $m$ samples.

\end{definition}

    

    



\item The following is the sparse version of the classical variant as the vectors $\vec a$'s have Hamming weight exactly $\sigma$. This is the variant that is most often used in the literature, in cryptographic constructions. 

\begin{definition}[$\lpn_{\eta,n,m}^{= \sigma}$]
For a secret $\vec s \sim \F_2^n$, we are given oracle access to $m$ samples $(\vec a,b)$, where 
$\vec a \sim {\binom{[n]}{\sigma}}$ uniformly and 
$b = \langle \vec a, \vec s \rangle \oplus e$ with 
$e \sim \text{Ber}(\eta)$.  The goal of the problem is to recover the secret $\vec s$ using these $m$ samples.
 \end{definition}
 \item The following is very similar, but is much more convenient to work with for our reductions. In this variant, the vector $\vec a$ is a boolean vector where each co-ordinate is chosen independently such that the expected hamming weight is $\sigma$,\begin{definition}[$\lpn_{\eta,n,m}^{\approx \sigma}$]
  For a secret $\vec s \sim \F_2^n$, we are given oracle access to $m$ samples $(\vec a,b)$, where 
$\vec a \sim (\text{Ber}(\sigma/n))^n$  and 
$b = \langle \vec a, \vec s \rangle \oplus e$ with 
$e \sim \text{Ber}(\eta)$.  The goal of the problem is to recover the secret $\vec s$ using these $m$ samples.

    
\end{definition}
\item All our results will extend for general prime field $\F_q$, and so we define these variants below. We could have completed avoided talking about LPN for $\F_2$ and introduced all definitions and proofs for $\F_p$. The only reason we do this separately is that most work in the cryptographic literature considers LPN modulo $2$ and a more general statement might be less appealing.
\begin{definition}[$\lpn_{q,\eta,n,m}^{= \sigma}$]
Let $X_{q,n,\sigma}$ denote the set of all vectors in $\F_q^n$, which have hamming weight exactly $\sigma$. For a secret $\vec s \sim \F_q^n$, we are given oracle access to $m$ samples $(\vec a,b)$, where 
$\vec a \sim X_{q,n,\sigma}$ uniformly and 
$b = \langle \vec a, \vec s \rangle + e$ with 
$e \sim \text{Ber}_q(\eta)$.  The goal of the problem is to recover the secret $\vec s$ using these $m$ samples.

 \end{definition}
\begin{definition}[$\lpn_{q,\eta,n,m}^{\approx \sigma}$]
  For a secret $\vec s \sim \F_q^n$, we are given oracle access to $m$ samples $(\vec a,b)$, where 
$\vec a \sim \paren{\Ber_q \paren{\frac{q \cdot \sigma}{(q-1) \cdot n}}}^n$  and 
$b = \langle \vec a, \vec s \rangle + e$ with 
$e \sim \text{Ber}_q(\eta)$. The goal of the problem is to recover the secret $\vec s$ using these $m$ samples.
\end{definition}
Again, here, the distribution of each co-ordinate of $\vec{a}$ is chosen to get the expected hamming weight of $\vec a$ to be $\sigma$.
 \end{enumerate}

 An algorithm $\mathcal{A}$ solves \textup{Search}-$\lpn$ with success probability $\alpha$ if,
\[\Pr_{\vec s,\vec A,\vec e}[\mathcal{ A}\paren{\vec A, \vec A\cdot\vec s + \vec e}=\vec s]\geq \alpha.\]
 

\label{sec:decode}


\subsection{Decodability of \texorpdfstring{$\lpn$}{LPN}}
In this section we prove that as long as $\eta$ is small enough, for a sample $(\vec A, \vec b = \vec A \vec s + \vec e) \sim \lpn_{q,\eta,n,m}$, the only vector $\vec v$ that satisfies $\norm{\vec b - \vec A \vec v}_0 \leq 2\eta m$ is the true secret $\vec v = \vec s$. Our approach is to first establish that a random code $\vec A \sim \F_q^{m \times n}$ has large distance, and then show that the noise vector $\vec e$ has small Hamming weight with high probability. Combining these facts via a union bound yields the desired conclusion. The following lemma establishes the large-distance property of a random code.

\begin{lemma}[GV-Bound]\label{lemma:GV} \cite{guruswami2012essential} Let $\vec A \sim \F_q^{m \times n}$, and let $\operatorname{dist}(\vec A)$ denote the minimum Hamming weight of a nonzero codeword in the linear code generated by $\vec A$. If $\eta <\dfrac{1}{4}(1-1/q)$ and $m>\paren{\frac{2}{1-H_q(4\eta)}}\cdot n$, then $\operatorname{dist }(\vec A)\leq 4{\eta}\cdot m$ with probability at most $q^{-n}$.  \end{lemma}

\begin{proof}
Since the code generated by $\vec A$ is simply $\operatorname{Img}(\vec A)$, i.e., the image of $\vec A$ when treated as a linear map from $\F_q^n$ to $\F_q^m$. We bound the probability that the image of $\vec A$ contains a non-zero vector of Hamming weight at most $4\eta m$.  
Let $B_q(r,m)$ denote the Hamming ball of radius $r$ in $\F_q^m$. Then,
\begin{align*}
    \Pr_{\vec A}\!\left[\operatorname{dist}(\vec A)\leq 4\eta \cdot m\right]
    &\leq \sum_{\vec{u} \in B_q\big(4\eta m,m\big) \setminus \{\vec 0\}}
           \Pr_{\vec A}\!\left[\vec{u} \in \operatorname{Img}(\vec A)\right] ,\\[6pt]
    &\leq \frac{\big|B_q\big(4\eta m,m\big)\big|}{q^{m-n}}, \\[6pt]
    &\leq q^{-\big((1-H_q(4\eta)) \cdot m - n\big)}, \\[6pt]
    &\leq q^{-n}.
\end{align*}
\qed\end{proof}
The following lemma is a straightforward application of Chernoff bounds, showing that the noise vector $\vec e$ has small Hamming weight with overwhelming probability.
\begin{lemma}\label{lemma:error}
    For $\vec e \sim \big(\text{Ber}(\eta)\big)^m$, the event 
    $\norm{\vec e}_0 > 2\eta m$ occurs with probability at most $e^{-\eta m /3}$. Also, for $\vec e \sim \big(\text{Ber}_q(\eta)\big)^m$,  the event 
    $\norm{\vec e}_0 > 2\eta m$ occurs with probability at most $e^{-\eta m /3}$.
\end{lemma}
\begin{proof}
It suffices to prove this statement only for $\vec e \sim \big(\text{Ber}(\eta)\big)^m$. The second statement follows from the fact that 
$$Pr[Ber_q(\eta) = 0] \geq Pr[Ber(\eta) = 0]\;. $$
Let $X_i$ be the indicator random variable for the event $\vec{e}_i \neq 0$, and set $X = \sum_{i=1}^m X_i = \|\vec{e}\|_0$.  
Then $\mathbb{E}[X] = \eta m$, and the $X_i$’s are independent.  
Applying the Chernoff bound with $\delta=1$, we obtain
\[
\Pr\!\left[ X \ge 2\eta m \right]
= \Pr\!\left[ X \ge (1+\delta)\eta m \right]
\;\leq\; \exp\!\left(-\tfrac{\delta^2}{2+\delta}\cdot \eta m\right)
= \exp\!\left(-\tfrac{1}{3}\eta m\right).
\]
\qed\end{proof}
Combining the two statements, we get the following conclusion.
\begin{corollary}\label{corollary:decode}
    For $\eta <\dfrac{1}{4}(1-1/q)$ and  $m>\paren{\frac{2}{1-H_q(4\eta)}}\cdot n$, the $\lpn_{q,\eta,n,m}$ sample $\paren{\vec A,\vec b= \vec A\cdot \vec s + \vec e}$ is such that, only the secret vector $\vec v=\vec s$ satisfies $\norm{\vec b-\vec A\vec v}_0 \leq {2\eta}\cdot m$, with probability $1 - q^{-n} - e^{-\eta m /3}$. 
\end{corollary}
\begin{proof}
    Using \cref{lemma:GV} we say can say that the distance of the code generated by the matrix $\vec A$ is at least $4 \eta m$ with probability at least $(1-q^{-n})$ and using  \cref{lemma:error}, $\norm{\vec b- \vec A \cdot \vec s}_0< 2\eta \cdot m$ with probability at least $(1-e^{-\eta m/3}) $. Hence by applying union bound we can conclude that with probability at least $1 - q^{-n} - e^{-\eta m /3}$, the distance of the code generated by the matrix $\vec A$ is at least $4 \eta m$, and $\norm{\vec b- \vec A \cdot \vec s}_0< 2\eta \cdot m$. Thus, with probability $1 - q^{-n} - e^{-\eta m /3}$, for any $\vec{v}\neq \vec s$, the Hamming distance of $\vec{A} \vec{v}$ from $\vec{A} \vec{s} + \vec{e}$ is more than $4 \eta m - 2 \eta m = 2\eta m$. The result follows.
 \qed\end{proof}

\subsection{Decodability of Sparse-\texorpdfstring{$\lpn$}{LPN}}
In this section we prove that for a sample $(\vec A, \vec b = \vec A \vec s + \vec e) \sim \lpn^{\approx \sigma}_{q,\eta,n,m}$, the only candidate vector $\vec v$ whose encoding $\vec A \vec v$ lies close to $\vec b$ is the true secret $\vec v = \vec s$. Our proof follows the same high-level strategy as in the standard $\lpn$ setting: first, we establish that a random code $\vec A$ has large distance; next, we argue that the noise vector $\vec e$ has small Hamming weight with overwhelming probability; finally, we combine these ingredients via a union bound to obtain the desired uniqueness guarantee.

Since we are no longer in the uniform sampling regime, a key technical step is to lower bound the probability $\Pr_{\vec a}[\langle \vec a, \vec u \rangle \neq 0]$ for nonzero $\vec u$, both in the binary and in the $q$-ary setting. The following lemmas establish these lower bounds, which in turn will help us conclude that the code generated by $\vec A$ has large distance, with high probability.We note that similar proofs appear in the literature; see, for example, \cite{sparse_lpn_algo}.
\begin{lemma}\label{lem:ber_inner_product}
    For $\vec a \sim \paren{\Ber\paren{\sigma/n}}^n$ and a non-zero vector $\vec u\in \F_2^n$ such that $\norm{\vec u}_0=\ell$, we have 
    \[\Pr_{\vec a}\left[\langle \vec a, \vec u\rangle \neq 0\right] =\dfrac{1-\paren{1- 2\sigma/n}^\ell}{2}\geq ~\dfrac{\sigma}{n}.\]
\end{lemma}
\begin{proof}
    If $\langle \vec u , \vec a \rangle=0$, this implies that $\vec u $ and $\vec a$ are simultaneously $1$ only at even many coordinates. Hence given $ \norm{\vec u}_0= \ell$, \begin{align*}
        \Pr_{\vec a}\left[\langle \vec u, \vec a\rangle \neq 0\right] &= 
        \sum_{i \text{ is odd}}{\binom{\ell}{i}} \paren{\frac{\sigma}{n}}^{i}  \paren{1-\frac{\sigma}{n}}^{\ell-i}, \\
        & = \dfrac{1-\paren{1- 2\sigma/n}^\ell}{2}.
    \end{align*}
\qed\end{proof}
\begin{lemma}

\label{general_support}
    For $\vec a \sim \paren{\Ber_q\paren{\frac{q\sigma}{(q-1)n}}}^n$ and a non-zero vector $\vec u\in \F_q^n$ such that $\norm{\vec u}_0=\ell$, we have 
    \[\Pr_{\vec a}\left[\langle \vec a, \vec u\rangle \neq 0\right]\;=\; \left(1 - \frac{1}{q}\right)\left(1 - \left(1 - \frac{q\sigma}{(q-1)n}\right)^{\ell}\right)\geq ~\dfrac{\sigma}{n}.
\] 
\end{lemma}
\begin{proof}
    Let $p_\ell$ be the probability that  $\Pr_{\vec a}\left[\langle \vec a, \vec u\rangle= 0\right]$, when $\norm{\vec u}_0= \ell$. Notice that $p_\ell$ satisfies the following recurrence relation where $\alpha = \frac{q\sigma}{(q-1)n}$,
    \[p_\ell = \paren{1 -\alpha + \frac{\alpha}{q}}\cdot p_{\ell-1} + \paren{\frac{\alpha}{q}} \cdot(1- p_{\ell -1}). \] Solving this recurrence with $p_0=1$ gives, 
    \[p_\ell = (1 - \alpha)^\ell + \frac{1}{q} \cdot \paren{1- (1-\alpha)^\ell} .\]
    Hence we get the required claim.
\qed\end{proof}
Observe that as $\ell $ increases, $\Pr[\langle\vec u, \vec a\rangle\neq 0]$, increases as well. Hence $\Pr[\langle \vec u , \vec a \rangle \neq 0] \geq (\sigma/n)$, for all non-zero vectors $\vec u$, in both cases. 
\begin{lemma} \label{decode_sparse}
    Let $m \geq 48 \cdot \left({n^3}/{\sigma^2}\right)$, with $\eta < \tfrac{1}{8}$, and $c = (0.25 - \eta)$.
Then, with probability at least $\paren{1-\frac{1}{2^{(2n-1)}}}$, for an instance $(\vec A, \vec b=\vec A \vec s +\vec e)$ of $\lpn^{\approx \sigma}_{\eta,n,m}$, the secret vector $\vec v = \vec s$ is the unique vector such that,
 \[\norm{\vec b -\vec A \vec v}_0<(\eta + c \cdot \frac{\sigma}{n}) \cdot m.\]
     
 \end{lemma}
 \begin{proof}

 If we sample $\vec A \sim (\Ber(\sigma/n))^{m \times n}$ and $\e \sim (\Ber(\eta))^m$, then for a given $\vec s \in \F_2^n$ and $\vec v \in \F_2^n$ such that $\norm{\vec s + \vec v}_0 = \ell>0$, we have
 \[\mathop{\E}[\norm{(\vec A  \vec s +\vec e) -(\vec A  \vec s)}_0] =\mathop{\mathbb E}[\norm{  \vec e}_0]= \eta m,\]
 Similarly, if we compute the expected distance between $\vec b = \paren{\vec A \vec s+ \vec e}$ and $\vec A \vec v$ we get,
 \begin{align*}
    & \mathop{\E}[\norm{(\vec A  \vec s +\vec e) - (\vec A  \vec v)}_0] \\&= \sum_{i=1}^{m}\Pr \left[ b_i \neq \langle \vec a_i, \vec v\rangle \right],\\
     &= m \cdot \paren{\Pr [\langle \vec a _i , \vec s + \vec v \rangle \neq e_i]},\\
     &= m \cdot \paren{ \eta \cdot \Pr[\langle \vec a_i, \vec s + \vec v \rangle = 0] +(1-\eta) \cdot \Pr[\langle \vec a_i, \vec s + \vec v \rangle \neq  0]},\\
     &=  m \cdot \paren{\eta + (1-2 \eta)\cdot \Pr[\langle \vec a_i, \vec s + \vec v \rangle \neq  0]}.
 \end{align*}
 
 Since, $c<(0.5-2\eta)$ and for $p:= \Pr[\langle \vec a_i, \vec s + \vec v \rangle \neq  0]$. If $\norm{\vec s + \vec v}_0>0$ then $p\geq \sigma/n$ using \cref{lem:ber_inner_product}, which implies $\paren{(1-2\eta)\cdot p -c \cdot \frac{\sigma}{n}}\geq p/2$. Hence we bound the following probability using Chernoff, \begin{align*}& \Pr_{\vec e, \vec A}[\norm{(\vec A\vec s + \vec e) -(\vec A \vec v)}_0 <(\eta + c \cdot \frac{\sigma}{n}) \cdot m ]\\ &\leq \exp\paren {-\paren{1-\frac{\eta + c\cdot \frac{\sigma}{n}}{\paren{\eta +(1-2\eta)\cdot p} } }^2 \cdot \frac{{\paren{\eta +(1-2\eta)\cdot p}}\cdot m }{2} },
 \\&\leq \exp \paren{-\frac{\paren{(1-2 \eta)\cdot p- c\cdot \frac{\sigma}{n}}^2}{\eta +(1-2\eta)\cdot p }\cdot \frac{m}{2}},\\
 & \leq \exp \paren{-\frac{{p}^2}{\eta +(1-2\eta)\cdot p }\cdot \frac{m}{8}},\\
 &\leq  \exp \paren{-{{p}^2}\cdot \frac{m}{8}}.
 \end{align*}

 If $\norm{\vec s + \vec v}_0>0$, then $p\geq 
 \sigma/n$ this implies,
 \[\Pr_{\vec e, \vec A}\left[\norm{(\vec A\vec s + \vec e) -(\vec A \vec v)}_0<(\eta + c \cdot \frac{\sigma}{n}) \cdot m\right] \leq \exp \paren{- \dfrac{\sigma^2 m}{8n^2}}.\]
 Now if $m \geq 48 \cdot \paren{{ n^3}/{\sigma^2} }$, then by union bound,
 \begin{align*}&\Pr\left[ \exists ~ \vec v \neq \vec s ~{\mid} ~ \norm{(\vec A\vec s + \vec e) -(\vec A \vec v)}_0<(\eta + c \cdot \frac{\sigma}{n}) \cdot m\right]\\& \leq 2^{n} \cdot \exp \paren{- \dfrac{\sigma^2 m}{8n^2}}\leq 4^{-n}. \end{align*}
 We will again use Chernoff's bound and the fact that $\eta<1/8, c>1/8$ in order to get the following bound,
 \begin{align*}
     \Pr[\norm{\vec e}_0 > (\eta +c \cdot \frac{\sigma}{n})\cdot m]< \exp(-\frac{c^2 \sigma^2}{3\eta n^2}\cdot m)\leq 4^{-n}.
 \end{align*}
 Hence, by union bound only the secret vector $\vec v=\vec s, $ satisfies $\norm{\vec b +\vec A \vec v}_0<(\eta + c \cdot \frac{\sigma}{n}) \cdot m$ with probability $\paren{1-\frac{1}{2^{(2n-1)}}}$.
 
  \qed\end{proof}
  Following exactly the same proof as \cref{decode_sparse} along with \cref{general_support} we can obtain a similar conclusion for $\lpn_{q,\eta,n,m}^{\approx \sigma}$.
\begin{lemma} \label{decode_sparse_q}
    Let $m \geq 48 \cdot \left({n^3}/{\sigma^2}\right) \cdot \log q$, with $\eta < \tfrac{1}{8}$ and $c = (0.25 - \eta)$.
Then, with probability at least $\paren{1-\frac{1}{2^{(2n-1)}}}$, for an instance $(\vec A, \vec b= \vec A \vec s +\vec e)$ of $\lpn^{\approx \sigma}_{q,\eta,n,m}$, the secret vector $\vec v = \vec s$ is the unique vector such that,
 \[\norm{\vec b -\vec A \vec v}_0<\paren{\paren{\paren{1- \frac{1}{q}}\cdot\eta }+ c \cdot \frac{\sigma}{n}} \cdot m.\]
     
 \end{lemma}

\begin{proof}
     If we sample $\vec A \sim \Ber_q\paren{\frac{q\sigma}{(q-1)n}}^{m \times n}$ and $\e \sim (\Ber_q(\eta))^m$, then for a given $\vec s \in \F_2^n$ and $\vec v \in \F_q^n$ such that $\norm{\vec v-\vec s}_0 = \ell>0$, we have
 \[\mathop{\E}[\norm{(\vec A  \vec s +\vec e) -(\vec A  \vec s)}_0] =\mathop{\mathbb E}[\norm{  \vec e}_0]= \paren{1- \frac{1}{q}}\cdot \eta m,\]
 Let $\vec u = (\vec v - \vec s) \neq 0$, now if we compute the expected distance between $\vec b = \paren{\vec A \vec s+ \vec e}$ and $\vec A \vec v$ we get,
 \begin{align*}
    & \mathop{\E}[\norm{(\vec A  \vec s +\vec e) - (\vec A  \vec v)}_0],\\ &= \sum_{i=1}^{m}\Pr \left[ b_i \neq \langle \vec a_i, \vec v\rangle \right],\\
     &= m \cdot \paren{\Pr [\langle \vec a _i , \vec u \rangle \neq e_i]},\\
     &= m \cdot \paren{\paren{1- \frac{1}{q}}\eta  \cdot \Pr[\langle \vec a_i, \vec u\rangle = 0] +\paren{1- \frac{\eta}{q}} \cdot \Pr[\langle \vec a_i, \vec u \rangle \neq  0]},\\
     &\geq  m \cdot \paren{\paren{1- \frac{1}{q}}\cdot\eta + \left(1-2 \paren{1- \frac{1}{q}}\eta\right)\cdot \Pr[\langle \vec a_i, \vec u \rangle \neq  0]}.
 \end{align*}

 Using \cref{general_support} we get,
 \[\mathop{\E}[\norm{(\vec A  \vec s +\vec e) - (\vec A  \vec v)}_0] \geq  m \cdot \paren{\paren{1- \frac{1}{q}}\cdot\eta + \left(1-2 \paren{1- \frac{1}{q}}\eta\right)\cdot \dfrac{s}{n}}.\]

 From this point, the proof proceeds analogously to \cref{decode_sparse}. The only difference is that, when applying the union bound over all vectors, we now sum over $q^n$ possibilities instead of $2^n$, introducing an additional factor of $\log q$ in the lower bound for $m$.
\qed\end{proof}

\section{Hardness self-amplification from \texorpdfstring{\cite{HS23}}{HS23}}
\label{sec:hs23}

One important conclusion from \cite{HS23} is that, the simple \emph{direct product reduction} is a powerful tool for hardness amplification. In this section, we provide a combinatorial exposition of some of their results. We emphasise that we do not claim any originality in the results.

\subsection{Direct Product Reduction}

Let $(f,\mu)$ be a distributional problem where $f:X \to \bit^*$ is a function defined on domain $X = \mathrm{supp}(\mu)$, $\mu$ is the input distribution. The \emph{$k$-wise direct product} of $(f,\mu)$ is the distributional problem $(f^{\times k},\mu^{\times k})$ defined as follows, for any $x_1,\ldots,x_k \in X$:
\begin{align*}
    &f^{\times k}(x_1, \ldots, x_k) = (f(x_1), \ldots, f(x_k))\; , \\
    &\mu^{\times k}(x_1,\ldots,x_k) = \prod_{i=1}^k \mu(x_i)\; .
\end{align*}
In other words, $\mu^{\times k}$ is simply the product distribution of $k$ independent copies of $\mu$, and input space of $(f^{\times k},\mu^{\times k})$ is $Y := X^k$.

Suppose we are given oracle access to a solver $\mathcal{O}$ for $(f^{\times k},\mu^{\times k})$. We define the direct product reduction $\mathcal{R}^\mathcal{O}$, which is a randomized reduction that uses $\mathcal{O}$ to solve $(f,\mu)$ as follows:

\begin{tcolorbox}[title=Reduction $\mathcal{R}^\mathcal{O}$, colback=white]
On input $x \in X$ and internal randomness $r$:
\begin{enumerate}
    \item For all $j \in [k]$, sample $x_j \sim \mu$ independently.
    \item Sample $i \in [k]$ uniformly at random. 
    \item Replace the sample $x_i$, with our input $x$ and form the tuple,
    \[
    {y} = (x_1,\ldots,x_{i-1},x,x_{i+1},\ldots,x_k) \in Y.
    \]
    \item Query the oracle $\mathcal{O}$ on input $y$.
    \item Let $(z_1,z_2, \cdots, z_i, \cdots  z_k)$ be the output of the $\mathcal{O}$, return $z_i$ as the solution for $x$.
\end{enumerate}
\end{tcolorbox}

\subsection{Hardness Amplification from Direct Product Reduction}
\begin{lemma}\label{lem: one_run}
    Let $(f, \mu)$ be a distributional problem. For $\eps, \delta, c \in [0,1]$ and positive integer $k$ satisfying that $2\exp(-kc^2\delta / 8)\leq c \eps$. Let $\mathcal{O}$ be an oracle that solves its $k$-wise direct product $(f^{\times k}, \mu^{\times k})$ with probability $\eps$ over its input distribution internal randomness. We have

    \[ \Pr_{x \sim \mu}\left[\Pr_{r}\left[\mathcal{R^O}(x,r) = f(x)\right] \geq (1-c)\eps \right] \geq 1-\delta \; .  \]
    
\end{lemma}
\begin{proof}
    Set $c' := c / 2, \eps' := c\eps / 2$. We use $G(x,r)$ to denote the instance $y \in Y$ generated in the reduction $\mathcal{R^O}(x,r)$.

    Define $H_X:= \{x:\Pr_{r}\left[\mathcal{R^O}(x,r) = f(x)\right] \leq (1-c)\eps \} $ to be the set of hard instances from $X$. Set 
    \[ \delta_1 := \Pr_{x \sim \mu}[x \in H_X] \;. \] 
    Note that it suffices to prove that $\delta_1 < \delta$ and let us assume towards contradiction that $\delta_1 \geq \delta$.
    
    We start with following claim, with the intuition that knowing $y$ does not reveal too much information about whether $x \in H_X$, since the reduction is ``mixing'' especially when $k$ is large:
    \begin{claim}
        \[ \Pr_{y \sim \mu^{\times k}}\left[ \Pr_{x \sim \mu, r}[x \in H_X | G(x,r) = y] \leq (1 - c') \delta_1 \right] \leq \eps' \; . \]
    \end{claim}
  
    \begin{proof}
        For any fixed $y = (x_1, \ldots, x_k)$, notice that 
        \begin{align*} 
            \Pr_{x \sim \mu, r}[x \in H_X | G(x,r) = y]  &= \sum_{j=1}^k \Pr_{i\sim [k]}[i = j] \cdot \mathbb{1}_{x_j \in H_X} \\  
            &= \frac{1}{k} \sum_{j=1}^k  \mathbb{1}_{x_j \in H_X} \; .
        \end{align*}
        The claim follows from applying Chernoff-Hoeffding's inequality by treating $\mathbb{1}_{x_j \in H_X}$ as a random variable about $y$ with expectation $\delta_1$.
        $\hfill \diamond$
    \qed\end{proof}        

    Next, set $w(y) := \Pr_{\mathcal{O}}[\mathcal{O}(y) = f(y)]$. By our assumption on $\mathcal{O}$, we have 
    \[ \Pr_{y \sim \mu^{\times k}, \mathcal{O}}[\mathcal{O}(y) = f(y)] = \sum_{y \in Y} \mu^{\times k}(y) w(y) \geq \eps \; . \]
    Set $S_Y:= \{y\in Y: \Pr_{x \sim \mu, r}[x \in H_X | G(x,r) = y] \leq (1 - c') \delta_1 \}$ be the set of instances from $Y$ that has relatively small contributions towards solving instances in $H_X$.

    Now consider two ways of evaluating the probability $$\Pr_{x\sim \mu, r} [x \in H_X, \mathcal{R^O}(x,r) = f(x)]$$.
    \begin{align*}
             & \Pr_{x\sim \mu, r} [x \in H_X, \mathcal{R^O}(x,r) = f(x)] \\
           = &\sum_{y \in Y} \mu^{\times k}(y) \Pr_{x \sim \mu, r}[x \in H_X | G(x,r) = y] \Pr_{\mathcal{O}}[\mathcal{O}(y) = f(y)] \\
        \geq & \sum_{y \in Y \setminus S_Y} \mu^{\times k}(y) \cdot \Pr_{x \sim \mu, r}[x \in H_X | G(x,r) = y] \cdot w(y) \\
        \geq & (1 - c')\delta_1 \sum_{y \in Y \setminus S_Y} \mu^{\times k}(y) \cdot w(y) \\
        \geq & (1 - c')\delta_1 \left(\sum_{y \in Y} \mu^{\times k}(y) w(y) - \sum_{y \in S_Y} \mu^{\times k}(y) w(y) \right) \\
        \geq & (1 - c')\delta_1 \left(\eps - \eps' \right) = (1 - c')(1 - \eps'/\eps)\eps \delta_1 \; .
    \end{align*}

    On the other hand, we have 

    \begin{align*}
             & \Pr_{x\sim \mu, r} [x \in H_X, \mathcal{R^O}(x,r) = f(x)] \\
        \leq & (1-c)\eps \Pr[x \in H_X] \\
        \leq & (1-c)\eps \delta_1 \; .
    \end{align*}    
    Combining the two inequalities we have $(1 - c')(1 - \eps'/\eps) \leq (1 - c)$ which is a contradiction by our choices of $c'$ and $\eps'$. \qed



\end{proof}
\begin{theorem}\label{product_rules}

    For every $\delta, \varepsilon$, $\gamma>0$ and $k\geq \frac{32}{\delta } \log{\paren{\frac{4}{\eps}}}$ if there exists,
    \begin{itemize}\item An algorithm $\mathcal A$ that takes input $y=(x_1,x_2, \cdots , x_k)$, runs in time $T(|y|)$ and solves the distributional problem $(f^{\times k},\mu ^k)$ with success probability $\varepsilon$.  
    \item An algorithm $\mathcal B$, which takes $(x,z)$ as input, runs in time $t(|x|)$ and verifies if $z=f(x)$, for all but $\vartheta$ fraction of $x$. \end{itemize}Then there exists an algorithm $\mathcal A'$, running in time $\mathcal{O}\paren{\log\paren{\tfrac{1}{\gamma}}\cdot \paren{\tfrac{T(kn)+ t(n)}{\varepsilon}}}$, for solving the distributional problem $(f,\mu)$, with success probability $(1-\delta -\gamma- \vartheta)$.
\end{theorem}
\begin{proof}
We will use the following algorithm $\mathcal A'$ in order to solve $(f,\mu)$,
\begin{tcolorbox}[
title= {An Algorithm $\mathcal A'$ for $(f,\mu)$}, colback=white]
Given an input $x$, we perform the following steps, for  $\ell=\left\lceil \frac{2 \log(\gamma^{-1})}{\varepsilon} \right\rceil$ iterations:
\begin{enumerate}
    \item Run the reduction $\mathcal{R^O}$ on $x$, where the oracle $\mathcal{O}$ is instantiated by the algorithm $\mathcal{A}$.
    \item Use $\mathcal{B}$ to verify whether $\mathcal{R}^{\mathcal{A}}(x) = f(x)$, if it is same as $f(x)$, then we output it and break out of the loop.
\end{enumerate}

\end{tcolorbox}
   Using \cref{lem: one_run} with $c=0.5$ we get, if $k\geq \frac{32}{\delta } \log{\paren{\frac{4}{\eps}}}$ then in one single iteration, for $(1-\delta)$ fraction of $x \sim \mu$, we can obtain $f(x)$ with probability $\geq 0.5\eps$. Moreover, for a $(1-\vartheta)$ fraction of $x\sim \mu$, the verifier succeeds in checking a proposed solution. Hence using union bound, for $(1-\delta -\vartheta)$ fraction of $x$, with probability $1 - (1 - 0.5\eps)^{2 \log(\gamma^{-1})/\eps} \geq 1 - \gamma$, we can obtain the solution for $x$ from the $\ell$ iterations. Using union bound, the overall success probability  of the above algorithm is at least $(1-\delta-\gamma-\vartheta)$.  Since in every iteration we make only one call to $\mathcal A$, the running time is clearly $\mathcal O(\ell \cdot (T(kn)+t(n))) = \mathcal O(\log(\gamma^{-1}) \eps^{-1}(T(kn)+t(n)))$. 
\qed\end{proof}

\section{Hardness Self Amplification for Search-\texorpdfstring{$\lpn$}{LPN}}
\label{sec:lpn_plain}

 In this section, we prove that direct-product self-reductions amplify average-case hardness for the standard $\lpn$.
Concretely, for any $\eps, \delta > 0$ and $k \in \Omega\!\left(\frac{1}{\delta} \log \frac{1}{\eps}\right)$, an algorithm that solves $\lpn_{q}$ in $n'$ dimensions with noise rate $\eta'$ and success probability $\eps$ can be turned into an algorithm that solves $\lpn_{q}$ in $n = n'/k$ dimensions with base noise rate $\eta > \eta'/k$, achieving success probability at least $1 - 2^{-\Omega(n)}$ on a $1-\delta$ fraction of instances.

\begin{theorem}\label{thm: main_gen_lpn}
    Let $m, n, k$ be positive integers with $m> \paren{\frac{2}{(1-H_q(4 \eta))}}\cdot n$. Let $\eps> 0$, $\delta > q^{-\Omega(n)}$, $\eta \in (0,1/8)$, and let $ \frac{192}{\delta} \cdot \log\!\left(\frac{8}{\eps}\right)< k <\frac{\epsilon}{8}\cdot \min \paren{q^{n},e^{\eta m /3}}$. There is a $\text{poly}(n, k, 1/\delta, 1/\eps)$-time reduction that, given an oracle that succeeds with probability $\eps$ for $\lpn_{q,\paren{{1-(1-\eta)^k}},kn,m}$, solves $\lpn_{q,\eta, n, m}$ with probability $1 - \delta$.
\end{theorem}

Let $\phi_k(\eta)=\paren{{1-(1-\eta)^k}}$.

We first show that a $k$-fold product of small $\lpn_q$ instances can be combined into a single larger $\lpn_q$ instance by concatenating the challenge matrices and summing the noisy responses, thereby matching the transformed noise rate $\phi_k(\eta)$. After establishing this reduction, we will apply the algorithm in \cref{product_rules} in order to conclude the proof of \cref{thm: main_gen_lpn}.

\begin{lemma}
\label{gen_lpn_lemma}
The product distributional problem $(\text{Search-}\lpn_{q,\eta,n,m})^k$ consists of $k$ independent samples from $\lpn_{q,\eta,n,m}$,
$
\big((\vec A_{1}, \vec{b}_1=\vec A_{1} \cdot \vec s_{1} + \vec e_{1}), \ldots, (\vec A_{k}, \vec{b}_k=\vec A_{k} \cdot \vec s_{k} + \vec e_{k})\big),
$
and the goal is to recover $\paren{\vec s_1, \cdots ,\vec s_k}$. Suppose there exists a $T(kn)$-time algorithm $\mathcal{A}$ that solves Search-$\lpn_{q,\phi_k(\eta),kn, m}$ with success probability $\eps$. Then there exists an algorithm $\Tilde{\mathcal A}$, running in time $T(kn)$, which solves $(\text{Search-}\lpn_{q,\eta, n,m})^k$ with success probability $\varepsilon$.
\end{lemma}
\begin{proof}
    Given an instance
    $
    y = \big((\vec A_{1}, \vec b_1=\vec A_{1} \cdot \vec s_{1} + \vec e_{1}), \ldots, (\vec A_{k}, \vec b_k=\vec A_{k} \cdot \vec s_{k} + \vec e_{k})\big)
    $
    of $(\text{Search-}\lpn_{q,\eta, n,m})^k$, we define $\paren{\vec A^*, \vec b^*}$ as follows,
    \[ 
    \vec{A}^* = \begin{pmatrix} \vec{A}_1 & \vec{A}_2 & \ldots& \vec{A}_k \end{pmatrix}, \qquad \vec b^*=\sum_{i=1}^{k}\vec b_i.
    \]
    Denote $\vec{s}^* = (\vec{s}_1, \vec{s}_{2},\ldots, \vec{s}_k)$ and $d=nk$. Notice that $\vec{A}^*$ is uniformly distributed over $\F_q^{m \times d}$ and $\vec{s}^*$ is uniformly distributed over $\F_q^{d}$.
    Moreover, let $\vec{e}^* =  \sum_{i\in [k] } \vec{e}_i$, then we have $\vec{b}^* = \vec{A}^*\vec{s}^* + \vec{e}^*$. Since $\vec{e}^*$ is the sum of $k$ independent random variables with distribution $(\Ber_q(\eta))^m$, by \cref{ber_q} it is exactly distributed as $\paren{\Ber_q\paren{{1-(1-\eta)^k}}}^m = \left(\Ber_q(\phi_k(\eta))\right)^m$. So we get
    $
    \vec A^* \sim \F_q^{m \times d}, \quad \vec{s}^*\sim \F_q^d, \quad \vec e^* \sim \left(\Ber_q(\phi_k(\eta))\right)^m, \quad \vec b^*= \vec A^*\cdot \vec s^* + \vec e^*.
    $
    Thus, we get an instance $(\vec A^*, \vec b^*)$ of $\text{Search-}\lpn_{q,\phi_k(\eta),kn,m}$. Running $\mathcal A$ on this instance yields $\vec s^*$ with success probability $\varepsilon$. Finally, note that $\vec s^* = (\vec s_1, \ldots, \vec s_k)$ is precisely the solution to the original instance $y$ as well. Therefore, we obtain an algorithm $\widetilde{\mathcal A}$ that solves $(\text{Search-}\lpn_{q,\eta,n,m})^k$ with success probability $\varepsilon$ in time $T(kn)$.
\qed\end{proof}

\begin{proof}[Proof of \cref{thm: main_gen_lpn}]
By \cref{gen_lpn_lemma}, the assumed $T(kn)$-time solver $\mathcal{A}$ for 
$\lpn_{q,\phi_k(\eta),kn,m}$ can be transformed into a $T(kn)$-time algorithm 
$\Tilde{\mathcal{A}}$ that solves the product problem 
$(\text{Search-}\lpn_{q,\eta,n,m})^k$ with success probability $\varepsilon$.

Let $\mu$ be the input distribution of $\lpn_{q,\eta,n,m}$. In order to apply \cref{product_rules}, we define the function $f:\F_q^{m\times n}\times \F_q^{m} \rightarrow \F_q^n \cup \{\perp \}$ as follows:
$
f(\vec A, \vec b) = \begin{cases}
    \vec{s} & \text{there exists a unique $\vec s$ such that } \|\vec b - \vec{As}\|_0 \leq 2\eta \cdot m \\
    \perp & \text{otherwise}
\end{cases} \; .
$

Notice that $\Tilde{\mathcal{A}}$ also solves $(f^{\times k}, \mu^{\times k})$ with success probability at least $\varepsilon - k\big(q^{-n} + e^{-\eta m /3}\big) \ge \varepsilon/2$. In particular, the extra error comes from the (non-decodable) event stipulated by \cref{corollary:decode}, which happens with probability at most $q^{-n} + e^{-\eta m /3}$ for each small instance.

In particular, here is a simple verification algorithm (required for \cref{product_rules}) for $(f, \mu)$: accept if and only if $\|\vec b - \vec{As}\|_0 \leq 2\eta \cdot m$.

Now applying \cref{product_rules} with the following parameters $(\delta' = \delta/3, \, \gamma' = \delta /3, \, \vartheta'= q^{-n} + e^{-\eta m /3})$, any $T(kn)$-time solver for 
$(f^{\times k}, \mu^{\times k})$ with success probability $\eps/2$ 
can be converted into an algorithm $\mathcal{A}'$ for $(f,\mu)$ with success probability $1 - 2\delta/3 - (q^{-n} + e^{-\eta m /3})$. Lastly, $\mathcal{A'}$ also solves $\lpn_{q,\eta,n,m}$ with the same success probability, except for a $(q^{-n} + e^{-\eta m /3})$ fraction of (not uniquely decodable) instances. The resulting algorithm runs in time $O((1/\varepsilon)\log(1/\delta)\cdot T(kn))$ and succeeds with probability at least $1-2\delta/3 - 2\big(q^{-n} + e^{-\eta m /3}\big) > 1 - \delta$ over the input distribution and its internal randomness as desired.
\qed\end{proof}


Using the above theorem with $q=2$, we obtain the following result for the binary case $\lpn_{\eta,n,m}$.
\begin{corollary}\label{thm: main_lpn}
    Let $m, n, k$ be positive integers with $m> \paren{\frac{2}{(1-H_2(4 \eta))}}\cdot n$. Let $\eps> 0$, $\delta > 2^{-\Omega(n)}$, $\eta \in (0,1/8)$, and let $ \frac{192}{\delta} \cdot \log\!\left(\frac{8}{\eps}\right)< k <\frac{\epsilon}{8}\cdot \min \paren{2^{n},e^{\eta m /3}}$. There is a $\text{poly}(n, k, 1/\delta, 1/\eps)$-time reduction that, given an oracle that solves $\lpn_{\frac{1-(1-2\eta)^k}{2}, nk, m}$  with success probability $\eps$, solves $\lpn_{\eta,n,m}$ with probability $1 - \delta$.
\end{corollary}

\begin{proof}
    Since $\lpn_{2,\eta',n,m}$ is identical to $\lpn_{\eta'/2,n,m}$, substituting $\eta'=2\eta$ in the base instance and $\eta'=1-(1-2\eta)^k$ in the scaled instance gives the claimed binary statement.
\end{proof}
\section{Hardness of Sparse Search-\texorpdfstring{$\lpn$}{LPN}}\label{sec:sparse_section}
\subsection{Hardness Self Amplification of Search-\texorpdfstring{$\lpn_{\eta,n,m}^{\approx \sigma}$}{Approx Sparse LPN}}
We will again follow the same template as \cref{sec:lpn_plain} in this section in order to prove hardness self amplification result about $\lpn^{\approx \sigma}_{q,\eta,n,m}$. 
\begin{theorem}\label{thm: main_sparse_lpn_q}
    Let $m, n, k, \sigma$ be positive integers with $m\geq 48\cdot (n^3\log q/\sigma^2)$. Let $\eps > 0, \delta > 6(2^{-(2n-1)}), \eta \in (0, 1/8)$, and let $\paren{\frac{192}{\delta} \cdot \log\!\left(\frac{8}{\epsilon}\right)}<k< 2^{(2n-2)} \cdot \varepsilon$. There is a $\poly(n, k, 1/\delta, 1/\epsilon)$-time reduction that, given an oracle which succeeds with probability $\eps$ for
    $\lpn_{q,\paren{{1-(1-\eta)^k}},kn, m}^{\approx k\sigma}$, solves $\lpn_{q,\eta,n,m}^{\approx \sigma}$ with success probability $1 - \delta$.
\end{theorem}

    

Let $\phi_k(\eta)=\paren{{{1-(1-\eta)^k}}{}}$.
We will now reduce $(\text{Search-}\lpn_{q,\eta,n,m}^{\approx \sigma})^k$ to  $\text{Search-}\lpn_{q,\phi_k(\eta),kn,m}^{\approx k\sigma}$. After establishing this reduction we will use the algorithm in \cref{product_rules} in order to get  \cref{thm: main_sparse_lpn_q}, 
\begin{lemma}
\label{sparse_lpn_lemma}
The product distributional problem $(\text{Search-}\lpn_{q,\eta,n,m}^{\approx \sigma})^k$ consists of $k$ independent samples from $\lpn_{q,\eta,n,m}^{\approx \sigma}$,
\[
\big((\vec A_{1}, \vec{b}_1=\vec A_{1} \cdot \vec s_{1} + \vec e_{1}), \ldots, (\vec A_{k}, \vec{b}_k=\vec A_{k} \cdot \vec s_{k} + \vec e_{k})\big),
\] and the goal is to recover $\paren{\vec s_1, \cdots ,\vec s_k}$. Suppose there exists a $T(kn)$-time algorithm $\mathcal{A}$ that solves Search-$\lpn_{q,\phi_k(\eta),kn, m}^{\approx k\sigma}$ with success probability $\eps$.

Then there exists an algorithm $\Tilde{\mathcal A}$, running in time $T(kn)$ which solves $(\text{Search-}\lpn_{q,\eta, n,m}^{\approx \sigma})^k$ with success probability $\varepsilon$.
\end{lemma}
\begin{proof}
    Given an instance, \[y = \big((\vec A_{1}, \vec b_1=\vec A_{1} \cdot \vec s_{1} + \vec e_{1}), \ldots, (\vec A_{k}, \vec b_k=\vec A_{k} \cdot \vec s_{k} + \vec e_{k})\big),\] of  $(\text{Search-}\lpn_{q,\eta, n,m}^{\approx \sigma})^k$, we define $\paren{\vec A^*, \vec b^*}$ as follows,
    \[\vec{A}^* = \begin{pmatrix} \vec{A}_1 & \vec{A}_2 & \ldots& \vec{A}_k \end{pmatrix}, \quad \vec b^*=\sum_{i=1}^{k}\vec b_i.\]
   Denote $\vec{s}^* = (\vec{s}_1, \vec{s}_{2},\ldots, \vec{s}_k)$ and $d=nk$. Notice that $\vec{A}^*$ is exactly distributed as $(\text{Ber}_q(q\sigma/(q-1)n))^{m \times d}=(\text{Ber}_q(kq\sigma/(q-1)d))^{m \times d}$ and $\vec{s}^*$ is uniformly distributed over $\F_q^{d}$.
    Moreover, let $\vec{e}^* =  \sum_{i\in [k] } \vec{e}_i$ , then we have $\vec{b}^* = \vec{A}^*\vec{s}^* + \vec{e}^*$. 
Since $\vec{e}^*$ is the sum of $k$ independent random variables with distribution $(\text{Ber}_q(\eta))^m$, by \cref{ber_q} it is exactly distributed as $\paren{\text{Ber}_q\paren{{1-(1-\eta)^k}{}}}^m = \left(\text{Ber}_q(\phi_k(\eta))\right)^m$. So we get,
 \[ \vec A^* \sim \text{Ber}_q\paren{\frac{k q\sigma}{(q-1)d}}^{m \times d}, ~ \vec{s}^*\sim \F_2^d, ~ \vec e^* \sim \left(\text{Ber}_q(\phi_k(\eta))\right)^m, ~ \vec b^*= \vec A^*\cdot \vec s^* + \vec e^*.\]
Thus, we have constructed an instance $(\vec A^*, \vec b^*)$ of $\text{Search-}\lpn_{q,\phi_k(\eta),kn,m}^{\approx k\sigma}$. Running $\mathcal A$ on this instance yields $\vec s^*$ with success probability $\varepsilon$. Finally, note that $\vec s^* = (\vec s_1, \ldots, \vec s_k)$ is precisely the solution to the original instance $y$ as well. Therefore, we obtain an algorithm $\widetilde{\mathcal A}$ that solves $(\text{Search-}\lpn_{q,\eta,n,m}^{\approx \sigma})^k$ with success probability $\varepsilon$ in time $T(kn)$.
    
\qed \end{proof}
\begin{proof}[Proof of \cref{thm: main_sparse_lpn_q}]
By \cref{sparse_lpn_lemma}, the assumed $T(kn)$-time solver $\mathcal{A}$ for 
$\lpn_{q,\phi_k(\eta),kn,m}^{\approx k\sigma}$ can be transformed into a $T(kn)$-time algorithm 
$\Tilde{\mathcal{A}}$ that solves the product problem 
$(\text{Search-}\lpn_{q,\eta,n,m}^{\approx \sigma})^k$ with success probability $\varepsilon$.

Let $\mu$ be the input distribution of $\lpn_{q,\eta,n,m}^{\approx \sigma}$. In order to apply \cref{product_rules}, we define the function $f:\F_q^{m\times n}\times \F_q^{m} \rightarrow \F_q^n \cup \{\perp \}$ as follows:
\[ f(\vec A, \vec b) = \begin{cases}
    \vec{s} & \text{$\exists$ unique $\vec s$ with } \|\vec b - \vec{As}\|_0 \leq(\eta(1-1/q) + (0.25 - \eta)\frac{\sigma}{n}) \cdot m \\
    \perp & \text{otherwise}
\end{cases}.\]

Notice that $\Tilde{\mathcal{A}}$ also solves $(f^{\times k}, \mu^{\times k})$ with success probability at least $\varepsilon - k(2^{-(2n-1)})$. In particular, the extra error comes from the (non-decodable) event stipulated by \cref{decode_sparse_q}, which happens with probability at most $2^{-(2n-1)}$ for each small instance.

In particular, there is a very simple verification algorithm (as required for \cref{product_rules}) for $(f, \mu)$: accept if and only if $\|\vec b - \vec{As}\|_0 \leq (\eta(1-1/q) + (0.25 - \eta)\frac{\sigma}{n}) \cdot m$.

Now applying \cref{product_rules} with the following parameters $(\delta' = \delta/3, \, \gamma' = \delta /3, \, \vartheta'= 2^{-(2n-1)})$, any $T(kn)$-time solver for 
$(f^{\times k}, \mu^{\times k}))$ with success probability $\varepsilon - k(2^{-(2n-1)}) > \eps/2$ 
can be converted into an algorithm $\mathcal{A}'$ for $(f,\mu)$ with success probability $1 - 2\delta/3 - (2^{-(2n-1)})$. Lastly, $\mathcal{A'}$ also solves $\lpn_{q,\eta,n,m}^{\approx \sigma}$ with the same success probability, except for a $(2^{-(2n-1)})$ fraction of (not uniquely decodable) instances. The resulting algorithm runs in time $O((1/\varepsilon)\log(1/\delta)\cdot T(kn))$ and succeeds with probability at least $1-2\delta/3 - 2(2^{-(2n-1)}) > 1 - \delta$ over the input distribution and its internal randomness as desired.


\qed \end{proof}

Using the above theorem with $q=2$ we get the following corollary.
\begin{corollary}\label{thm: main_sparse_lpn}
    Let $m, n, k, \sigma$ be positive integers with $m\geq 48\cdot (n^3/\sigma^2)$. Let $\eps > 0, \delta > 6(2^{-(2n-1)}), \eta \in (0, 1/8)$, and let $\paren{\frac{192}{\delta} \cdot \log\!\left(\frac{8}{\epsilon}\right)}<k< 2^{(2n-2)} \cdot \varepsilon$. There is a $\poly(n, k, 1/\delta, 1/\epsilon)$-time reduction that, given an oracle which succeeds with probability $\eps$ for $\lpn_{\paren{\frac{{1-(1-2\eta)^k}}{2}},kn, m}^{\approx k\sigma}$ , solves $\lpn_{\eta,n,m}^{\approx \sigma}$ with success probability $1 - \delta$.
\end{corollary}
\begin{proof}
    Since $\lpn_{2,\eta',n,m}^{\approx \sigma}$ is identical to $\lpn_{\eta'/2,n,m}^{\approx \sigma}$, substituting $\eta'=2\eta$ in the base instance and $\eta'=1-(1-2\eta)^k$ in the scaled instance gives the claimed binary statement.
\qed\end{proof}
\subsection{Reduction from Search-\texorpdfstring{$\lpn_{q,\eta,n,3 \sqrt{\sigma} \cdot m}^{\approx \sigma}$}{Approx Sparse LPN} to Search-\texorpdfstring{$\lpn_{q,\eta,n,m}^{=\sigma}$}{Normal Sparse LPN}}
In this section we will give a reduction from $\lpn^{\approx \sigma}_{q,\eta,n,3\sqrt{\sigma} \cdot m}$ to $\lpn^{= \sigma}_{q,\eta,n,m}$. We set $M := 3\sqrt{\sigma} \cdot m$. 

In particular, given an instance $(\vec A, \vec b= \vec A \cdot \vec s+ \vec e)$ of $\lpn_{q,\eta,n,M}^{\approx \sigma}$, we will output $\mathcal{R}\paren{\vec A, \vec b} \in \F^{m \times (n+1)} \cup \{\bot\}$, such that distribution of $\mathcal{R}\paren{\vec A, \vec b}$ is statistically close to the input distribution of $\lpn_{q,\eta,n,m}^{=\sigma}$. 

If $\mathcal{R}\paren{\vec A, \vec b} \neq \bot$, then let $\mathcal{R}\paren{\vec A, \vec b}$ be $(\vec A', \vec b')$

Moreover, the secret $\vec s$ for the original $\lpn_{q,\eta,n,M}^{\approx \sigma}$ instance will be the same as the secret $\vec s'$ of $(\vec A', \vec b')$.

Now, if an algorithm $\mathcal{A}$ can solve $\lpn_{q,\eta,n,m}^{=\sigma}$ with success probability $\alpha$ then given the above reduction it can solve $\lpn_{q,\eta,n,M}^{\approx \sigma}$ with success probability $\alpha - \kappa$, where $\kappa$ is the distance between the distribution induced by the reduction, and the input distribution of $\lpn_{q,\eta,n,m}^{=\sigma}$.
\subsubsection{The Reduction
}

Given an instance of $\lpn_{q,\eta,n,M}^{\approx \sigma}$ as,

$$\paren{\vec A = \begin{bmatrix}
\vec a_1 \\
\vec a_2 \\
\vdots \\
\vec a_{M}
\end{bmatrix}, ~ \vec b=\begin{bmatrix}
b_1 \\
b_2 \\
\vdots \\
b_{M}
\end{bmatrix}
  },$$
If the number of rows $\vec a$ with hamming weight exactly $\sigma$ is strictly less than $m$ then we output $\perp$. Otherwise, let $[\vec a_ {i_1}, \vec{a}_{i_2} \cdots \vec{a}_{i_m}]$ be the  first $m$ rows of $\vec{A}$ with hamming weight exactly $\sigma$. The reduction outputs  
$$\paren{\vec A' = \begin{bmatrix}
\vec a_{i_1} \\
\vec a_{i_2} \\
\vdots \\
\vec a_{i_m}
\end{bmatrix}, ~ \vec b'=\begin{bmatrix}
b_{i_1} \\
b_{i_2} \\
\vdots \\
b_{i_m}
\end{bmatrix}
  }.$$
   To establish that the distribution of $(\vec A', \vec b')$ is statistically close 
to that of the instances of $\lpn_{q,\eta,n,m}^{=\sigma}$, 
we first analyze the distribution of $\vec A'$.
We claim that the distribution of $\vec A'$ is statistically close to the 
uniform distribution over ${\binom{[n]}{\sigma}}^m$. 
Conditioned on the event that there are at least $m$ rows 
with Hamming weight exactly $\sigma$ (i.e., that the reduction doesn't output $\bot$), 
the distribution of $\vec A', \vec{b}'$ is distributed as $\lpn_{q,\eta,n,m}^{=\sigma}$ with the same secret vector $\vec{s}$.  
Let $\beta$ be the probability that $\vec A$ has strictly less than $m$ rows with Hamming weight $\sigma$, i.e., that $\mathcal{R}\paren{\vec A, \vec b} = \bot$. Thus, if the $\lpn_{q,\eta,n,m}^{=\sigma}$ oracle succeeds in finding the secret vector $\vec{s}$ with probability $\eps$, then the reduction is successful with probability 
\begin{align*}
&\Pr[\mathcal{R}\paren{\vec A, \vec b} \neq \bot] \cdot \Pr[ \text{ Reduction outputs } \vec s | \mathcal{R}\paren{\vec A, \vec b} \neq \bot] \\
&= \Pr[\mathcal{R}\paren{\vec A, \vec b} \neq \bot] \cdot \Pr[\lpn_{q,\eta,n,m}^{=\sigma} \text{ oracle outputs } \vec s] \\ &
= (1-\beta) \cdot \eps \\ &\ge \eps - \beta \;.
\end{align*}

Thus, to prove correctness of our reduction, we give an upper bound on $\beta$.

 \begin{lemma}
 \label{lem:bound_sparse_q}
Let $q$ be a prime power, and let $\sigma,n,m \in \mathbb{N}$ with
$M=\lceil 3\sqrt{\sigma}\,m \rceil$ and $m\geq 1200n$. Let
$\vec a_1,\ldots,\vec a_M$ be sampled independently as
\[
    \vec a_i \sim
    \left(\Ber_q\left(\frac{q\sigma}{(q-1)n}\right)\right)^n .
\]
Then, with probability at least $1-\frac{1}{2^n}$, there exist at least
$m$ rows among the $\vec a_i$'s whose Hamming weight is exactly $\sigma$.
\end{lemma}

\begin{proof}
Let $X_i$ be the indicator random variable for the event that
$\vec a_i$ has Hamming weight exactly $\sigma$. Put
\[
    \alpha := \frac{q\sigma}{(q-1)n}.
\]
For a single coordinate $Y\sim \Ber_q(\alpha)$, we have
\[
    \Pr[Y=0] = (1-\alpha)+\frac{\alpha}{q},
    \qquad
    \Pr[Y=v] = \frac{\alpha}{q}
    \quad \text{for every } v\in \F_q\setminus\{0\}.
\]
Hence
\[
    \Pr[Y\neq 0]
    = \frac{q-1}{q}\alpha
    = \frac{\sigma}{n}.
\]
Therefore the Hamming weight of $\vec a_i$ is distributed as
$\mathsf{Bin}(n,\sigma/n)$, and so
\[
    \Pr[X_i=1]
    =
    \binom n\sigma
    \left(\frac{\sigma}{n}\right)^\sigma
    \left(\frac{n-\sigma}{n}\right)^{n-\sigma}.
\]
Using \cref{lem:aprox_bion}, we get
\begin{align*}
    \Pr[X_i=1]
    &=
    \binom n\sigma
    \left(\frac{\sigma}{n}\right)^\sigma
    \left(\frac{n-\sigma}{n}\right)^{n-\sigma} \\
    &\geq
    \frac{1}{\sqrt{8\sigma(1-\tfrac{\sigma}{n})}}
    \cdot
    \frac{n^n}{\sigma^\sigma (n-\sigma)^{n-\sigma}}
    \cdot
    \left(\frac{\sigma}{n}\right)^\sigma
    \left(\frac{n-\sigma}{n}\right)^{n-\sigma} \\
    &=
    \frac{1}{\sqrt{8\sigma(1-\tfrac{\sigma}{n})}}
    \geq
    \frac{1}{\sqrt{8\sigma}} .
\end{align*}

Now let,
$
    X := \sum_{i=1}^M X_i
$
which implies that
$
    \mathbb E[X]
    \geq \frac{M}{\sqrt{8\sigma}}.
$
Since $M=\lceil 3\sqrt{\sigma}\,m\rceil$, we have
\[
    m \leq \frac{M}{3\sqrt{\sigma}}
    = \frac{M}{\sqrt{9\sigma}}
    =
    \sqrt{\frac{8}{9}}\cdot \frac{M}{\sqrt{8\sigma}}
    \leq
    \sqrt{\frac{8}{9}}\cdot \mathbb E[X].
\]
Thus, setting
$
    \delta := 1-\sqrt{\frac{8}{9}},
$
we get
$
    m \leq (1-\delta)\cdot \mathbb E[X].
$
By Chernoff's bound,
\begin{align*}
    \Pr[X<m]
    &\leq
    \Pr[X<(1-\delta)\mathbb E[X]] \\
    &\leq
    \exp\left(-\frac{\delta^2 \mathbb E[X]}{2}\right) \\
    &\leq
    \exp\left(-\frac{\delta^2 M}{2\sqrt{8\sigma}}\right) \\
    &\leq
    \exp\left(-\frac{3\delta^2 m}{4\sqrt{2}}\right).
\end{align*}
Since $m\geq 1200n$, the final expression is at most $e^{-n}$, and hence at
most $2^{-n}$. Therefore, with probability at least $1-2^{-n}$, there are
at least $m$ rows among the $\vec a_i$'s whose Hamming weight is exactly
$\sigma$.
\end{proof}

    This gives us the following theorem,
  \begin{theorem}\label{thm: approx_exact_q}
        Let $\sigma,n,m \in \mathbb{N}, M=\lceil 3\sqrt{\sigma}\,m \rceil$,  and $m\geq 1200n$. If there exists an algorithm $\mathcal A$ running in time $T(n)$, that solves $\lpn_{q,\eta,n,m }^{=\sigma}$ with success probability $\varepsilon$, then it can solve
        $\lpn_{q,\eta,n,M}^{\approx \sigma}$ with success probability $\paren{\varepsilon-\frac{1}{2^{n}}}$ in time $T(n)$.
    \end{theorem}

 Following exactly the same proof as above we obtain the following corollary for the binary field,
    \begin{corollary}\label{thm: approx_exact}
        Let $\sigma,n,m \in \mathbb{N}, M=\lceil 3\sqrt{\sigma}\,m \rceil$,  and $m\geq1200 n$. If there exists an algorithm $\mathcal A$ running in time $T(n)$, that solves $\lpn_{\eta,n,m }^{=\sigma}$ with success probability $\varepsilon$, then it can solve
        $\lpn_{\eta,n,M}^{\approx \sigma}$ with success probability $\paren{\varepsilon-\frac{1}{2^{n}}}$ in time $T(n)$.
    \end{corollary}
   


\subsection{Putting Everything Together}
By combining the reductions from \cref{thm: approx_exact} and \cref{thm: main_sparse_lpn}, we have the following theorem:

\begin{theorem}
    Let $m, n, k, \sigma$ be positive integers with $m\geq \max(48\cdot (n^3/\sigma^2),1200n)$ and set $M=\lceil 3\sqrt{\sigma}\,m \rceil$. Let $\eps > 0, \delta > 6(2^{-(2n-1)}), \eta \in (0, 1/8)$, and $k$ satisfies $\paren{\frac{192}{\delta} \cdot \log\!\left(\frac{16}{\epsilon}\right)}<k< 2^{(2n-2)} \cdot \varepsilon$. There is a $\poly(n, k, 1/\delta, 1/\epsilon)$-time reduction that, given an oracle for $\lpn_{\paren{\frac{{1-(1-2\eta)^k}}{2}},kn, m}^{= k\sigma}$ that succeeds with probability $\eps$, solves $\lpn_{\eta,n,M}^{\approx \sigma}$ with success probability $1 - \delta$.
\end{theorem}

Similarly by combining the reductions from \cref{thm: approx_exact_q} and \cref{thm: main_sparse_lpn_q} we can get the same theorem and conclusions for $\lpn_{q,\eta,n,m}^{=\sigma}$.

\begin{theorem}
    Let $m, n, k, \sigma$ be positive integers with $m\geq \max(\tfrac{48 n^3\log q}{\sigma^2},1200n)$ and set $M=\lceil 3\sqrt{\sigma}\,m \rceil$. Let $\eps > 0, \delta > 6(2^{-(2n-1)}), \eta \in (0, 1/8)$, and $k$ satisfies $\paren{\frac{192}{\delta} \cdot \log\!\left(\frac{16}{\epsilon}\right)}<k< 2^{(2n-2)} \cdot \varepsilon$. There is a $\poly(n, k, 1/\delta, 1/\epsilon)$-time reduction that, given an oracle for $\lpn_{q,\paren{{{1-(1-\eta)^k}}},kn, m}^{= k\sigma}$ that succeeds with probability $\eps$, solves $\lpn_{q,\eta,n,M}^{\approx \sigma}$ with success probability $1 - \delta$.
\end{theorem}

\bibliographystyle{alpha}
\bibliography{references}
\appendix
\section{Proof of \texorpdfstring{\cref{lem:aprox_bion}}{Corollary regarding Strilings}} \label{appendix}
\begin{proof}[Proof of \cref{lem:aprox_bion}]

Since
\[
\binom{n}{s}=\frac{n!}{s!(n-s)!}.
\]
For the \emph{lower} bound, we use the \emph{lower} bound for $n!$ and the \emph{upper}  bounds for $s!$ and $(n-s)!$ from \cref{Rob} to obtain,
\[
\binom{n}{s}
\;\ge\;
\frac{
\sqrt{2\pi}\, n^{\,n+\frac12} e^{-n+\frac{1}{12n+1}}
}{
\Bigl(\sqrt{2\pi}\, s^{\,s+\frac12} e^{-s+\frac{1}{12s}}\Bigr)
\Bigl(\sqrt{2\pi}\, (n-s)^{\,n-s+\frac12} e^{-(n-s)+\frac{1}{12(n-s)}}\Bigr)
}.
\]
Canceling the common $e^{-n}$ with $e^{-k}e^{-(n-k)}$ and collecting powers we get,
\begin{align*}
\binom{n}{s}
&\;\ge\;
\frac{1}{\sqrt{2 \pi s(1-\tfrac{s}{n})}}
  \cdot 
  \frac{n^n}{s^s (n-s)^{\,n-s}}
\cdot
\exp\!\Bigl(
\frac{1}{12n+1}-\frac{1}{12s}-\frac{1}{12(n-s)}
\Bigr),\\
&\;\ge\ \frac{1}{\sqrt{2 \pi s(1-\tfrac{s}{n})}}
  \cdot 
  \frac{n^n}{s^s (n-s)^{\,n-s}}
\cdot e^{-1/6} \geq \frac{1}{\sqrt{ 8 s(1-\tfrac{s}{n})}}
  \cdot 
  \frac{n^n}{s^s (n-s)^{\,n-s}}
;
\end{align*}

which yields the stated lower bound. 
For the \emph{upper} bound, we use the \emph{upper} bound for $n!$ and the \emph{lower}  bounds for $s!$ and $(n-s)!$ from \cref{Rob} to obtain,
\[
\binom{n}{s}
\;\le\;
\frac{
\sqrt{2\pi}\, n^{\,n+\frac12} e^{-n+\frac{1}{12n}}
}{
\Bigl(\sqrt{2\pi}\, s^{\,s+\frac12} e^{-s+\frac{1}{12s+1}}\Bigr)
\Bigl(\sqrt{2\pi}\, (n-s)^{\,n-s+\frac12} e^{-(n-s)+\frac{1}{12(n-s)+1}}\Bigr)
},
\]
and again after cancellations and using the fact $\paren{\frac{1}{12n}-\frac{1}{12s+1}-\frac{1}{12(n-s)+1}}<0$ we get,
\begin{align*}
\binom{n}{s}
&\;\le\;
\frac{1}{\sqrt{2 \pi s(1-\tfrac{s}{n})}}
  \cdot 
  \frac{n^n}{s^s (n-s)^{\,n-s}}
\cdot
\exp\!\Bigl(
\frac{1}{12n}-\frac{1}{12s+1}-\frac{1}{12(n-s)+1}
\Bigr),\\
&\;\le\; \frac{1}{\sqrt{2 \pi s(1-\tfrac{s}{n})}}
  \cdot 
  \frac{n^n}{s^s (n-s)^{\,n-s}}
\cdot
\\
\end{align*}
\end{proof}

\end{document}